\documentclass{article}

\usepackage{arxiv}

\usepackage[utf8]{inputenc}
\usepackage[T1]{fontenc}
\usepackage{hyperref}
\usepackage{url}
\usepackage{booktabs}
\usepackage{amsfonts}
\usepackage{amsmath}
\usepackage{amssymb}
\usepackage{amsthm}
\usepackage{nicefrac}
\usepackage{microtype}
\usepackage{graphicx}
\usepackage{doi}
\usepackage{listings}
\usepackage{xcolor}

\newtheorem{definition}{Definition}

\newtheorem{proposition}{Proposition}
\newtheorem{remark}{Remark}

\definecolor{codebg}{RGB}{248,248,248}
\definecolor{codekw}{RGB}{0,0,180}
\definecolor{codestr}{RGB}{0,128,0}
\definecolor{codecom}{RGB}{120,120,120}

\lstdefinestyle{pystyle}{
  backgroundcolor=\color{codebg},
  basicstyle=\ttfamily\footnotesize,
  keywordstyle=\color{codekw}\bfseries,
  stringstyle=\color{codestr},
  commentstyle=\color{codecom}\itshape,
  breaklines=true,
  captionpos=b,
  showstringspaces=false,
  language=Python,
  frame=single,
  framerule=0.3pt,
  rulecolor=\color{black!20},
  xleftmargin=0.5em,
  xrightmargin=0.5em,
  aboveskip=0.6em,
  belowskip=0.6em,
}
\title{State Twins: An Off-Chain Substrate for Agentic Reasoning\\
over Decentralized Finance Protocols}

\author{
  \href{https://orcid.org/0009-0001-6031-4066}{\includegraphics[scale=0.06]{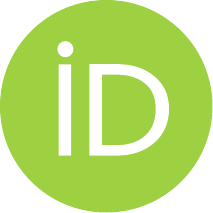}\hspace{1mm}Ian C. Moore, PhD}\\
  Principal, DeFiMind \\
  \texttt{imoore@defimind.ai}
}

\date{May 12, 2026}

\renewcommand{\shorttitle}{State Twins for Agentic DeFi}

\hypersetup{
pdftitle={State Twins: An Off-Chain Substrate for Agentic Reasoning over Decentralized Finance Protocols},
pdfsubject={cs.DC},
pdfauthor={Ian C. Moore},
pdfkeywords={decentralized finance, automated market maker, digital twin, agentic systems,
             large language models, model context protocol, state-space models,
             constant function market makers, Uniswap, Balancer, simulation},
}

\begin{document}
\maketitle

\begin{abstract}
We introduce the \emph{State Twin}: a typed, in-memory, replayable replica of an on-chain automated market maker (AMM) pool that serves as a substrate for agentic reasoning over decentralized finance (DeFi) protocols. Agentic DeFi stacks today couple reasoning to chain time, since every "what if?" query incurs a new RPC read or a real transaction, so the agent's effective action space is bounded by block confirmation latency and gas. We argue this coupling is a structural problem rather than a performance one, and that the missing layer is an off-chain substrate that preserves the protocol's exact mathematics while admitting the operations on-chain state cannot: forking, replay, branching, counterfactual rollout. We formalize each AMM family (Uniswap V2, V3, Balancer, Stableswap) as a discrete-time controlled dynamical system, prove a quantitative fidelity bound on the divergence between twin and chain, and give the open architecture used in DeFiPy v2, an open-source Python toolkit that ships the State Twin substrate and a reference Model Context Protocol server exposing typed analytical primitives as LLM tools. The same primitive (i.e., one Python class, one calling pattern) serves a notebook quant, a backtest, and an LLM agent without modification. We close with a fork-and-evaluate worked example: a single live RPC read seeds N independent in-memory twins under distinct price-shock scenarios, in sub-second wall-clock time. The contribution is the substrate, not a particular agent, which is what the specification of \emph{what an agentic DeFi substrate must look like}.
\end{abstract}

\keywords{decentralized finance \and automated market makers
\and digital twin \and agentic systems \and large language models
\and model context protocol \and state-space models \and Python SDK}

% ============================================================
\section{Introduction}
\label{sec:intro}
% ============================================================

The dominant pattern for agentic DeFi systems today wires a reasoning model
(typically a large language model) directly to chain reads and chain writes,
mediated by a thin RPC wrapper or a vendor data API. We will call this
architecture \emph{reactive}: that is, every step the agent takes either returns a
fresh on-chain query result or commits a transaction. There is no layer in
between. The agent's effective action space is therefore bounded by chain
time (i.e., block confirmation latency, gas, sequential transaction
ordering), and its reasoning is bounded by what one round-trip to the chain
can answer. Counterfactual queries, multi-scenario evaluation, and what an
analyst would casually call \emph{thinking before acting} are not first-class
operations in the reactive architecture; they are workflows the architecture
silently disallows.

This paper argues that reactive is the wrong default, and that agentic DeFi
specifically has a clean structural fix. AMM pools are deterministic state
machines: their on-chain state at block $k$, together with a transaction
$u_k$, exactly determines their state at block $k{+}1$ via an invariant
that is fully known and compactly expressible in software. There is no
hidden state, no privileged side channel, no non-deterministic adversary.
The implication is that an off-chain replica can be constructed that is
\emph{mathematically faithful} to chain reality at any pinned block: exact
in the rational arithmetic of the protocol invariants, and exact up to a
known additive rounding bound under the deployed fixed-point arithmetic
\cite{TG25}. Primitives consuming such a replica are equivalent
in their outputs whether the replica was sourced from a live RPC read, a
synthetic recipe, a fork node, or a historical archive.

We call this replica the \emph{State Twin}. The contribution of this paper
is the formalization of the State Twin as the missing substrate for
agentic DeFi, which admits a clean provider/builder factorization that survives across protocol families. This enables a class of agent workflows (i.e., fork, simulate, score,
decide) that the reactive architecture cannot support at all.

\paragraph{Contributions.} This paper makes the following contributions:

\begin{enumerate}
\item We formalize each major AMM family as a discrete-time controlled
      dynamical system and show that this state-space framing is the natural
      mathematical structure for agentic DeFi reasoning
      (Section~\ref{sec:state-space}).

\item We define the \emph{State Twin} as the canonical software realization
      of the state-space view, give a precise abstraction
      (provider $\to$ snapshot $\to$ builder $\to$ twin $\to$ primitive)
      that decouples state source from state shape and realizes the substrate-access path required for counterfactual reasoning, and prove a quantitative
      \emph{State Twin Fidelity} bound on the divergence between twin and
      on-chain trajectories under the discretized constant-product update
      (Section~\ref{sec:twin}).

\item We expose the State Twin pattern through DeFiPy v2, an open-source
      Python toolkit that ships 7 core operational primitives (Swap, Join, AddLiquidity, etc.) 
      and 21 typed analytical primitives built on them, a State Twin substrate (\texttt{MockProvider}, 
      \texttt{LiveProvider}), and a reference Model Context Protocol server exposing 10 of those
      primitives as LLM tools (Section~\ref{sec:defipy}).

\item We present a \emph{fork-and-evaluate} pattern: a single chain read
      seeds $N$ independent twins, each evaluated under a distinct
      scenario, producing a recommendation distribution in sub-second
      wall clock time (Section~\ref{sec:fork-evaluate}).

\item We discuss the substrate-vs-product distinction underlying the
      design: why the agentic surface is deliberately read-only, why the
      math layer remains LLM-free, and why this is the architecture most
      likely to compound across protocols (Section~\ref{sec:discussion}).
\end{enumerate}

The State Twin, as a named pattern with an explicit
provider/builder split, is designed specifically for LLM-native agentic
reasoning over AMM state. The reactive architecture is implicit in essentially every
production agentic-DeFi system we surveyed during this work.

% ============================================================
\section{Preliminaries: AMMs as Controlled Dynamical Systems}
\label{sec:state-space}
% ============================================================

We adopt a state-space view of AMM protocols. The motivation is twofold:
(1) the dynamical-systems formulation is the natural language for
counterfactual queries (e.g., ``what is the state if I apply this trade?''),
which is what an agent that thinks before it acts most needs to ask;
(2) the formulation makes explicit that an off-chain replica is
mathematically well-posed, since the state-transition map is a pure function
of state and input, with no hidden randomness on the protocol side.

\subsection{General Form}

Let $\mathcal{X}$ denote the protocol-specific state space, $\mathcal{U}$
the input space (swaps, mints, burns, fee accrual), and $\mathcal{W}$
the exogenous-disturbance space (e.g., oracle prices, external trades observed
through the same pool). An AMM pool evolves as the discrete-time system:
\begin{equation}
\label{eq:state-space}
\boxed{\quad x_{k+1} \;=\; f(x_k,\, u_k,\, w_k), \qquad y_k \;=\; h(x_k) \quad,}
\end{equation}
where $f: \mathcal{X} \times \mathcal{U} \times \mathcal{W} \to \mathcal{X}$
is the transition map (the protocol invariant in solver form) and
$h: \mathcal{X} \to \mathcal{Y}$ is the observation map (e.g.\ spot price,
TVL, position value). Expression \eqref{eq:state-space} is the discrete-time controlled state-space form introduced by Kalman \cite{Kalman60} and standard in control and estimation theory; we adopt it here as the deterministic specialization for AMM protocols, with $f$ fixed to the deployed protocol invariant and $w_k$ treated as known exogenous disturbances, as substrate for off-chain agentic reasoning.

The underlying state-update structure of \eqref{eq:state-space} has appeared
in several adjacent formalisms in prior work, each suited to a specific
purpose: stochastic state-space for adaptive mechanism design \cite{NKV24},
where $u_k$ and $w_k$ are promoted to random variables and Kalman filtering
recovers a hidden external price; constant-function characterization for
theoretical analysis \cite{AC20}, where $f$ preserves a trading function
$\varphi$ over the level set of valid reserves; priced timed automata for
verification \cite{TG25}, where the state in $\mathcal{X}$ includes clock
valuations and transitions are edge-labeled; and dependently-typed
functional transitions for theorem-proving correctness \cite{PB24}. The
treatments differ in vocabulary and proof toolkit, convex-optimization,
Bayesian inference, model-checking, and theorem-proving, but the underlying
state-update structure is the same. Our deterministic specialization of
\eqref{eq:state-space} occupies the substrate-construction slot in this
landscape: $f$ is known, $w_k$ is treated as a known exogenous input rather
than a noise term, and the resulting object is used not for inference,
verification, or theoretical characterization but as a typed, replayable,
forkable off-chain replica of the on-chain pool.

The protocol invariant $\mathcal{I}: \mathcal{X} \to \mathbb{R}_{\geq 0}$
admits a precise per-transition behaviour that depends on whether the
underlying arithmetic is real-valued or discretized. The protocol-specific instantiations of $f$, $h$, and $\mathcal{I}$ for Uniswap V2, V3, Balancer, and Stableswap that we present in the subsections below follow the hand-derived AMM math treatment of \cite{Moo25-Book}; we summarize each compactly here and refer the reader there for complete derivations. Under real-valued arithmetic, fee-free swaps preserve $\mathcal{I}$ exactly and fee-accruing swaps weakly increase it. Under the deployed fixed-point arithmetic, with floor rounding on the output leg, the invariant is preserved up to a controlled, one-sided additive perturbation:
\begin{equation}
\label{eq:invariant-preservation}
\mathcal{I}(x_{k}) - \rho_k \;\leq\; \mathcal{I}(x_{k+1}) \;\leq\; \mathcal{I}(x_k) + \phi_k,
\end{equation}
where $\phi_k \geq 0$ is the fee accrual at step $k$ and $\rho_k \geq 0$ is
the rounding slack, bounded by the relevant reserve at step $k$
\cite{TG25}. We return to a quantitative form of \eqref{eq:invariant-preservation}
in the V2 subsection below.

This formulation cleanly separates three things: \emph{state} (the
quantities the chain stores), \emph{input} (transactions an agent or
external party submits), and \emph{observation} (the analytics the agent
or analyst extracts). Reactive architectures conflate all three by
forcing every observation $y_k$ to be a fresh chain read; the State Twin
materializes $x_k$ off-chain so that observations and counterfactual
transitions $f(x_k, u', w')$ become local computation.

\subsection{Uniswap V2: Constant Product}

For a Uniswap V2 pool with reserves $(r_0, r_1)$ and fee $\phi$ (e.g.\
$\phi = 0.003$), the state is $x_k = (r_0^{(k)}, r_1^{(k)})$ and the
invariant is
\begin{equation}
\label{eq:v2-invariant}
\mathcal{I}_{\mathrm{V2}}(x) \;=\; r_0 \cdot r_1 \;=\; k.
\end{equation}
A swap $u_k = (\Delta_0^{\mathrm{in}}, 0)$ trading token $0$ for token $1$
induces, in real-valued arithmetic, the transition
\begin{equation}
\label{eq:v2-transition}
r_0^{(k+1)} = r_0^{(k)} + \Delta_0^{\mathrm{in}}, \qquad
r_1^{(k+1)} = r_1^{(k)} - \frac{(1-\phi)\,\Delta_0^{\mathrm{in}}\,r_1^{(k)}}{r_0^{(k)} + (1-\phi)\,\Delta_0^{\mathrm{in}}}.
\end{equation}
The deployed protocol implements this in fixed-point integer arithmetic
with floor rounding on the output leg. Tranquilli and Gupta
\cite{TG25} formalize this discretization for the fee-free case as
\begin{equation}
\label{eq:v2-discretized}
r_0^{(k+1)} = r_0^{(k)} + \Delta_0^{\mathrm{in}}, \qquad
r_1^{(k+1)} = \left\lfloor \frac{r_0^{(k)} \cdot r_1^{(k)}}{r_0^{(k+1)}} \right\rfloor,
\end{equation}
and prove the following bound on the resulting product
$K' := r_0^{(k+1)} \cdot r_1^{(k+1)}$ relative to $K := r_0^{(k)} \cdot r_1^{(k)}$:
\begin{equation}
\label{eq:tg-bound}
K - r_0^{(k+1)} \;\leq\; K' \;\leq\; K \qquad \Rightarrow \qquad
|K' - K| \leq r_0^{(k+1)}.
\end{equation}
For a sequence of $n$ such swaps with reserve bound
$r_0^{(j)} \leq B$ for all $j$, telescoping
\eqref{eq:tg-bound} yields the multi-step bound
\begin{equation}
\label{eq:tg-multistep}
|K_n - K_0| \;\leq\; nB.
\end{equation}
This is the rigorous form of \eqref{eq:invariant-preservation} for V2:
the slack $\rho_k$ at step $k$ is bounded by $r_0^{(k+1)}$, and the
cumulative slack over $n$ steps is bounded by $nB$. The observation map
$h$ exposes the spot price $p = r_1/r_0$, the TVL in either numeraire,
and the LP value derived from a position's fractional ownership of the
reserves; on real-valued state, these are closed-form functions of $x_k$
alone, with no chain calls and no oracle dependencies.

\subsection{Uniswap V3: Concentrated Liquidity over Ticks}

V3 partitions price space into discrete ticks; liquidity $L_i$ is allocated
to the range $[T_i, T_{i+1})$. Inside a single tick range with active
liquidity $L_{\mathrm{act}}$ and price $P$, the protocol maintains
\emph{virtual reserves} $(r_0^{\mathrm{virt}}, r_1^{\mathrm{virt}})$
defined so that $L_{\mathrm{act}} = \sqrt{r_0^{\mathrm{virt}} \cdot r_1^{\mathrm{virt}}}$;
swaps inside the active tick evolve $(r_0^{\mathrm{virt}}, r_1^{\mathrm{virt}})$
exactly as in V2, with cross-tick transitions accounted for by re-pricing
virtual reserves at each tick boundary \cite{Adams21,TG25}. The state
$x_k$ includes the current $\sqrt{P}$, active tick, total liquidity, and
the tick-bitmap. For \emph{active-liquidity} primitives (those that operate
while the pool is inside a single tick range), the V3 transition is
structurally identical to V2 with the discretized rule \eqref{eq:v2-discretized}
applied to the virtual reserves; the rounding bound \eqref{eq:tg-bound}
applies verbatim \cite{TG25}. This is exactly why a V3 pool, once
linearized at the active tick, exposes the same primitive surface as V2
against the State Twin (see Section~\ref{sec:fork-evaluate}).

\subsection{Balancer: Weighted Geometric Mean}

Balancer 2-asset pools generalize V2 with token weights $(w_0, w_1)$,
$w_0 + w_1 = 1$. The invariant is
\begin{equation}
\label{eq:balancer-invariant}
\mathcal{I}_{\mathrm{Bal}}(x) \;=\; r_0^{w_0} \cdot r_1^{w_1}.
\end{equation}
A 50/50 Balancer pool reduces (up to scaling) to V2; an 80/20 pool exposes
asymmetric impermanent loss, which the State Twin makes available as a
local primitive call rather than as a chain query.

\subsection{Stableswap: Amplified Invariant}

Curve-style stableswap pools blend the constant-sum and constant-product
invariants via an amplification coefficient $A$. For two assets,
\begin{equation}
\label{eq:stableswap-invariant}
A \cdot n^n \cdot \!\sum_i\! r_i \;+\; D \;=\; A \cdot n^n \cdot D \;+\; \frac{D^{n+1}}{n^n \prod_i r_i},
\end{equation}
with $D$ implicitly defined and recovered via Newton iteration on $x_k$.
Stableswap admits a closed-form $\varepsilon \leftrightarrow \delta$
relationship between depeg magnitude and LP value loss; depeg-risk
primitives in DeFiPy invert this relationship over the State Twin to
produce risk metrics without any chain dependency at evaluation time.

\subsection{The State-Space Implication}

The unifying observation across all four families is that
\eqref{eq:state-space} is exactly evaluable off-chain, given $x_k$ and a
candidate $(u_k, w_k)$, because $f$ is a pure function. This is the
mathematical justification for the State Twin: an off-chain replica that
preserves $x_k$ is sufficient for any analytical workflow, including
arbitrary counterfactual rollouts, that does not require canonical
on-chain settlement. The reactive architecture's implicit assumption,
that every observation must come from the chain, is a software choice,
not a mathematical necessity.

\begin{remark}[Forced Factorization and the Substrate Necessity]
\label{rem:factorization}
Let $\chi: \mathcal{P} \to \mathcal{X}$ denote the chain-read map carrying
a pool identifier $p \in \mathcal{P}$ to its current on-chain state. In
a reactive architecture, every state-dependent query is forced to factor
through $\chi$: state access reduces to $\chi$ itself; observation
reduces to $h \circ \chi$; transition evaluation reduces to chain
settlement followed by $\chi$. This factorization is not required by the
underlying mathematics of \eqref{eq:state-space}, in which $h$ and $f$
are pure functions on $\mathcal{X}$ and
$\mathcal{X} \times \mathcal{U} \times \mathcal{W}$ respectively. The
forced factorization renders three classes of queries operationally
inaccessible: \emph{counterfactual transitions} $f(\hat{x}, u', w')$ for
inputs $(u', w')$ not committed to the chain; \emph{ensemble queries}
$\{h(\hat{x}^{(s)})\}_{s \in S}$ over fork families
$\{\hat{x}^{(s)}\}_{s \in S}$; and \emph{trajectory queries}
$\{h(\hat{x}_k)\}_{k=0}^{n}$ over rollouts of arbitrary depth $n$.

A corollary, sketched: any architecture supporting counterfactual queries
must contain a local realization of $\mathcal{X}$ admitting direct
evaluation of $h$ and $f$ without reference to $\chi$. The argument is
brief: a counterfactual transition $f(\hat{x}, u', w')$ with $u'$
uncommitted corresponds to no value of $\chi(p)$ for any
$p \in \mathcal{P}$, so it cannot factor through $\chi$, so its
evaluation must be mediated by some $\mathcal{X}$-valued object accessible
independently of $\chi$. The State Twin (Section~\ref{sec:twin}) is the
canonical such object; we will develop the full characterization elsewhere.
\end{remark}

% ============================================================
\section{The State Twin}
\label{sec:twin}
% ============================================================

\subsection{Definition}

\begin{definition}[State Twin]
\label{def:state-twin}
A \emph{State Twin} of an AMM pool $\mathcal{P}$ at block $k$ is a typed
in-memory object $\widehat{x}_k$ together with a set of operations
$\Omega = \{f, h, \mathrm{clone}\}$ such that, for some additive
fidelity slack $\varepsilon_k \geq 0$:
\begin{itemize}
\item[(T1)] \emph{State fidelity.} $\widehat{x}_k$ is structurally
            isomorphic to the on-chain state $x_k$ of $\mathcal{P}$ at
            block $k$ under the natural decimal-adjusted mapping, with
            componentwise discrepancy bounded by $\varepsilon_k$.

\item[(T2)] \emph{Observational equivalence.} For every observable
            $y$ computable from state alone via the observation map
            $h$,
            \[
            |\,h(\widehat{x}_k) - h(x_k)\,| \;\leq\; L_h \cdot \varepsilon_k
            \]
            where $L_h$ is the (local) Lipschitz constant of $h$ on the
            relevant state region. On real-valued arithmetic
            $\varepsilon_k = 0$ and equivalence is exact.

\item[(T3)] \emph{Transitional equivalence.} For every transition input
            $(u, w) \in \mathcal{U} \times \mathcal{W}$,
            \[
            |\,f(\widehat{x}_k, u, w) - f(x_k, u, w)\,| \;\leq\; \varepsilon_{k+1},
            \]
            where $\varepsilon_{k+1}$ is the propagated slack after one
            transition. For the discretized constant-product rule, the
            per-step contribution to $\varepsilon_{k+1}$ is bounded by
            the relevant reserve, per \eqref{eq:tg-bound}.

\item[(T4)] \emph{Independent forking.} $\mathrm{clone}(\widehat{x}_k)$
            produces an independent twin $\widehat{x}_k'$ such that
            subsequent operations on $\widehat{x}_k'$ leave
            $\widehat{x}_k$ unchanged.
\end{itemize}
\end{definition}

The four properties capture, respectively: \emph{state fidelity} (T1),
\emph{observational equivalence} (T2), \emph{transitional equivalence}
(T3), and \emph{independent forking} (T4). The Lipschitz constant $L_h$ in T2 invokes the standard quantitative bound of Lipschitz \cite{Lipschitz1876} on how state perturbations propagate through a continuous map, in the form standard to modern control theory \cite{Khalil02}; for AMM observables such as spot price, TVL, and position value, $L_h$ is small on realistic state regions, so observation slack is dominated by the state slack $\varepsilon_k$ itself. T4 is what reactive architectures fundamentally cannot provide: an on-chain pool admits exactly one trajectory per epoch, while a State Twin admits arbitrarily
many. Properties T2 and T3 specialize the classical observation and transition equation faithfulness conditions of state-space control theory \cite{Kalman60, Khalil02} to the State Twin setting; T1 grounds them in initial-state fidelity, and T4 adds the structural property, independent forking, that distinguishes a deterministic-substrate twin from a classical sensor/plant model. 

The slack $\varepsilon_k$ is zero on real-valued state (the case relevant
to most analytical primitives, which operate on rationals or floats with
sufficient precision) and is bounded explicitly when the twin is meant to
mirror the deployed fixed-point arithmetic of the chain. The next
proposition makes the bound concrete for the constant-product family.

\begin{proposition}[State Twin Fidelity]
\label{prop:fidelity}
Let $\widehat{x}_0 = x_0$ be a State Twin built from an exact chain
snapshot at block $k_0$, evolving under the real-valued constant-product
update rule \eqref{eq:v2-transition} with $\phi = 0$. Suppose the on-chain
pool, evolving under the discretized update rule \eqref{eq:v2-discretized},
and the twin are driven by the same sequence of fee-free swap inputs
$(\delta_1, \dots, \delta_n)$ with the on-chain reserve bounded by
$r_0^{(j)} \leq B$ for all $1 \leq j \leq n$. Then the constant-product
observable $h_K(x) := r_0 \cdot r_1$ satisfies
\begin{equation}
\label{eq:fidelity-bound}
|\,h_K(\widehat{x}_n) - h_K(x_n)\,| \;\leq\; nB,
\end{equation}
in raw token units, where the discrepancy is zero in the absence of
rounding (i.e.\ when both systems are evaluated on real-valued state)
and otherwise bounded by the single-step result of Tranquilli and Gupta
\cite{TG25} applied telescopically.
\end{proposition}

\begin{proof}
With $\widehat{x}_0 = x_0$ and identical inputs, the real-valued twin
preserves the constant-product invariant exactly under fee-free swaps,
so $h_K(\widehat{x}_n) = K_0$ for all $n \geq 0$. The on-chain pool,
evolving under floor-rounded discretized arithmetic
\eqref{eq:v2-discretized}, satisfies the per-step bound \eqref{eq:tg-bound};
telescoping across $n$ swaps gives $|h_K(x_n) - K_0| \leq nB$ via
\cite[Theorem~1]{TG25}. The triangle inequality then yields
$|h_K(\widehat{x}_n) - h_K(x_n)| \leq nB$. \qed
\end{proof}

\begin{remark}[Realistic magnitudes]
\label{rem:magnitudes}
The bound \eqref{eq:fidelity-bound} is tight in the worst case but loose
in practice. For a typical mainnet pool with reserves on the order of
$10^{18}$--$10^{24}$ minimal token units and $n$ on the order of tens of
swaps in a fork-and-evaluate sweep, the relative slack $nB / K_0$ is on
the order of $10^{-10}$ or smaller \cite{TG25}, well below any threshold
relevant to analytical decision-making. The substrate is therefore
\emph{economically exact} on realistic instances, even where it is only
\emph{provably $\varepsilon$-close} in the worst case. The bound is the
worst-case statement under TG's \cite{TG25} formalism; tighter bounds are
available under additional structural assumptions (bounded swap-size
relative to reserves, specific arithmetic-mode choices, or distributional
assumptions on the swap sequence) and represent natural follow-up work
that the substrate's deterministic foundation makes accessible.
\end{remark}

\subsection{The Provider / Builder Factorization}

The State Twin definition (T1-T4) specifies what the substrate must satisfy as an object in $\mathcal{X}$; Remark~\ref{rem:factorization} establishes that such an object must exist accessible independently of the chain-read map $\chi$ for any architecture supporting counterfactual queries. We now construct it. The construction factors into two functions: a \emph{Provider} $\pi: \mathcal{P}_{\text{ext}} \to \mathcal{S}$ from an external identifier space $\mathcal{P}_{\text{ext}}$ (synthetic recipe names, chain addresses, CSV row keys) to a typed canonical \emph{Snapshot} $\mathcal{S}$, and a \emph{Builder} $\beta: \mathcal{S} \to \mathcal{X}$ that lifts the snapshot into the state space as a typed in-memory object satisfying Definition~\ref{def:state-twin}. Their composition $\beta \circ \pi: \mathcal{P}_{\text{ext}} \to \mathcal{X}$ realizes the substrate-access path that Remark~\ref{rem:factorization}'s corollary requires, replacing $\chi$ as the universal state-access map without losing fidelity (Proposition~\ref{prop:fidelity}).

The DeFiPy reference implementation realizes each formal object as a
named software artifact:

\begin{itemize}
\item The external identifier space $\mathcal{P}_{\text{ext}}$ is the
      \texttt{pool\_id} string type, with provider-specific semantics
      (synthetic recipe name for \texttt{MockProvider}, chain
      \texttt{"<protocol>:<address>"} for \texttt{LiveProvider}, or
      any encoding a custom provider chooses).

\item The Provider $\pi$ is the \texttt{StateTwinProvider} abstract base
      class with a single \texttt{snapshot(pool\_id)} method;
      concrete subclasses (\texttt{MockProvider}, \texttt{LiveProvider},
      \texttt{CSVProvider}) know about \emph{sources}.

\item The Snapshot $\mathcal{S}$ is the typed \texttt{PoolSnapshot}
      dataclass family (\texttt{V2PoolSnapshot}, \texttt{V3PoolSnapshot},
      \texttt{BalancerPoolSnapshot}, \texttt{StableswapPoolSnapshot}),
      one per protocol family.

\item The Builder $\beta$ is the \texttt{StateTwinBuilder} dispatch
      that maps a \texttt{PoolSnapshot} to a concrete exchange object
      satisfying Definition~\ref{def:state-twin}; the builder dispatch
      knows about \emph{shapes}.
\end{itemize}

\paragraph{Why this factorization.} The factorization matters because the
Provider and Builder vary along fundamentally different axes. The
Provider varies with the \emph{source} of pool state: a synthetic
recipe, a live RPC read, a fork node, a CSV file, a historical archive.
The Builder varies with the \emph{shape} of pool state: which protocol
family the pool belongs to (V2, V3, Balancer, Stableswap) and how its
invariant is encoded. A naive design that merges the two produces
$|\mathrm{sources}| \times |\mathrm{protocols}|$ adapters; the
factorization reduces this to $|\mathrm{sources}| + |\mathrm{protocols}|$.
New sources extend the Provider axis; new protocols extend the Builder
axis; the Snapshot $\mathcal{S}$ is the canonical hand-off.

\paragraph{The State Twin pipeline.} An end-to-end agent or analyst flow is:
\[
\boxed{\;\;
\texttt{pool\_id} \;\xrightarrow{\;\textsc{Provider}\;}\;
\texttt{Snapshot} \;\xrightarrow{\;\textsc{Builder}\;}\;
\widehat{x}_k \;\xrightarrow{\;\textsc{Primitive}\;}\;
\texttt{Result}
\;\;}
\]
The right two arrows, builder and primitive, are the same regardless
of source. A primitive cannot tell whether its input came from a
synthetic mock or a live mainnet read, by design.

\subsection{Reference Implementation in DeFiPy}

DeFiPy v2 ships the abstraction in the \texttt{defipy.twin} module
\cite{DeFiPy26}. The provider abstract base class is small:

\begin{lstlisting}[caption={The \texttt{StateTwinProvider} abstract base class.}]
class StateTwinProvider(ABC):
    """Abstract source of pool snapshots."""

    @abstractmethod
    def snapshot(self, pool_id: str, **kwargs) -> PoolSnapshot:
        """Return a typed snapshot for the given pool identifier.

        pool_id semantics are provider-specific:
            MockProvider:   recipe name (e.g., "eth_dai_v2")
            LiveProvider:   "<protocol>:<address>"
            CustomProvider: any encoding the consumer chooses.
        """
        ...
\end{lstlisting}

Two concrete providers ship in v2:

\begin{itemize}
\item \texttt{MockProvider}: four canonical synthetic recipes (V2, V3,
      Balancer 50/50, Stableswap with $A{=}10$). Used in notebooks, tests,
      LLM demos, and any setting where ground-truth on-chain state is
      irrelevant or intentionally controlled.

\item \texttt{LiveProvider}: live mainnet reads for Uniswap V2 and V3
      via \texttt{web3.py}. Snapshots resolve \texttt{"latest"} to a
      concrete block number once, then pin every subsequent read to that
      block. V3 reads are batched via Multicall3 \cite{Multicall3} into a
      single round trip. Snapshots are decimal-adjusted to whole-token
      units to match \texttt{MockProvider}'s contract.
\end{itemize}

A custom provider is roughly twenty lines of code; the snapshot type is
the only thing it needs to produce. Listing~\ref{lst:csv-provider} shows
a CSV-backed provider in full.

\begin{lstlisting}[label={lst:csv-provider},caption={A custom State Twin provider that reads pool state from a CSV row. The same builder, primitive, and downstream agent code work unchanged.}]
from defipy.twin import StateTwinProvider, V2PoolSnapshot

class CSVProvider(StateTwinProvider):
    """Load V2 pool state from a per-pool CSV row."""

    def __init__(self, csv_path: str):
        self.rows = self._load(csv_path)

    def snapshot(self, pool_id: str, **kwargs) -> V2PoolSnapshot:
        row = self.rows[pool_id]
        return V2PoolSnapshot(
            pool_id     = pool_id,
            token0_name = row["token0"],
            token1_name = row["token1"],
            reserve0    = float(row["reserve0"]),
            reserve1    = float(row["reserve1"]),
        )
\end{lstlisting}

\subsection{Snapshot Enrichment and Chain Context}

Every snapshot from \texttt{LiveProvider} carries optional chain-context
fields: \texttt{block\_number}, \texttt{timestamp}, and \texttt{chain\_id}.
\texttt{MockProvider} leaves these fields as \texttt{None}: a synthetic
snapshot has no chain context, and inventing values would be a fidelity
violation. Consumers that need chain context (e.g., for cache keys, reorg
detection, multi-chain routing) branch on \texttt{None}; consumers that do
not, ignore it. This optional-enrichment pattern is what allows the same
primitive surface to serve notebook quants, backtests against historical
blocks, and LLM agents calling tools, without per-consumer adapters.

% ============================================================
\section{DeFiPy v2: An Agentic Python Toolkit}
\label{sec:defipy}
% ============================================================

DeFiPy v2 is the open-source reference implementation of the State Twin
pattern \cite{DeFiPy26}. It is published on PyPI and ships under three
install profiles: \texttt{pip install defipy} (core analytics, zero
web3/LLM dependencies), \texttt{defipy[chain]} (adds \texttt{LiveProvider}),
and \texttt{defipy[agentic]} (adds the MCP server SDK for LLM tool
exposure). The architectural commitment is that the math layer stays
deterministic, dependency-light, and LLM-free; LLM integration happens
strictly at the agent layer.

\subsection{The Primitive Interface}

Every analytical primitive in DeFiPy follows a three-line interface:

\begin{lstlisting}[caption={The primitive interface: stateless construction, computation at apply, typed dataclass return.}]
primitive = SomePrimitive()              # Stateless construction
result    = primitive.apply(lp, *args)   # Computation against twin
value     = result.<typed_field>         # Structured dataclass access
\end{lstlisting}

Three invariants hold across all 21 shipped primitives. \emph{Stateless
construction:} the primitive holds no state that affects the math; calling
\texttt{apply} twice with the same inputs returns the same answer.
\emph{Computation at apply:} all work happens in one call; no
subscriptions, no streams, no internal caches. \emph{Typed dataclass
return:} every primitive returns a specific result dataclass (e.g.,
\texttt{PositionAnalysis}, \texttt{PoolHealth}, \texttt{PriceMoveScenario})
with named fields, no raw tuples, no dicts, no stringly-typed
returns.

The interface is what allows the same primitive to be called from a
notebook by a quant, from a test by CI, and from an LLM via MCP, with
zero adapter code. Listing~\ref{lst:notebook-call} shows the canonical
notebook flow.

\begin{lstlisting}[label={lst:notebook-call},caption={Calling a primitive against a State Twin from a Jupyter notebook.}]
from defipy import AnalyzePosition
from defipy.twin import MockProvider, StateTwinBuilder

# 1. Source: pick a synthetic recipe
provider = MockProvider()
snapshot = provider.snapshot("eth_dai_v2")

# 2. Shape: build the twin
lp = StateTwinBuilder().build(snapshot)

# 3. Compute: run any primitive
result = AnalyzePosition().apply(
    lp,
    lp_init_amt = 1.0,
    entry_x_amt = 1000.0,
    entry_y_amt = 100000.0,
)

print(f"Diagnosis: {result.diagnosis}")
print(f"Net PnL:   {result.net_pnl:.4f}")
print(f"IL %:      {result.il_percentage:.4f}")
\end{lstlisting}

The same three lines work against a live RPC by replacing one import:

\begin{lstlisting}[label={lst:live-rpc-call},caption={The live-RPC analog of Listing~\ref{lst:notebook-call}: the primitive call is unchanged; only the provider import differs.}]
from defipy.twin import LiveProvider, StateTwinBuilder

provider = LiveProvider("https://eth-mainnet.example.com/v2/<key>")
snapshot = provider.snapshot(
    "uniswap_v2:0xB4e16d0168e52d35CaCD2c6185b44281Ec28C9Dc"
)
lp = StateTwinBuilder().build(snapshot)  # now WETH/USDC at chain head
\end{lstlisting}

\subsection{The 21 Primitives}

The shipped primitives span nine categories and four protocol families. Table~\ref{tab:primitives} summarizes the surface; all are typed-dataclass returning, all consume a State Twin via the same \texttt{lp} argument (the twin instance produced by the Builder, as in Listing~\ref{lst:notebook-call}), and all are exhaustively unit-tested.

\begin{table}[h]
\centering
\footnotesize
\begin{tabular}{lll}
\toprule
Category & Example primitive & Returns \\
\midrule
Position analysis & \texttt{AnalyzePosition} & \texttt{PositionAnalysis} \\
Price scenarios   & \texttt{SimulatePriceMove} & \texttt{PriceMoveScenario} \\
Pool health       & \texttt{CheckPoolHealth} & \texttt{PoolHealth} \\
Risk              & \texttt{DetectRugSignals}, \texttt{AssessDepegRisk} & \texttt{RugSignalReport}, \dots \\
Optimization      & \texttt{OptimalDepositSplit}, \texttt{EvaluateRebalance} & \texttt{DepositSplit}, \dots \\
Comparison        & \texttt{CompareFeeTiers}, \texttt{CompareProtocols} & \texttt{FeeTierComparison}, \dots \\
Execution         & \texttt{CalculateSlippage}, \texttt{DetectMEV} & \texttt{SlippageAnalysis}, \dots \\
Portfolio         & \texttt{AggregatePortfolio} & \texttt{PortfolioAggregate} \\
Break-even        & \texttt{FindBreakEvenPrice}, \texttt{FindBreakEvenTime} & \texttt{BreakEvenPoint} \\
\bottomrule
\end{tabular}
\caption{The nine analytical-primitive categories shipped in DeFiPy v2,
one representative primitive per category and the typed dataclass it
returns. All 21 primitives consume a State Twin via the same \texttt{lp}
argument and follow the interface in Listing~3.}
\label{tab:primitives}
\end{table}

\subsection{LLM Exposure via the Model Context Protocol}

A curated subset of 10 primitives is exposed as Model Context Protocol
(MCP) \cite{MCP24} tools through a reference server shipped in the
\texttt{defipy[agentic]} install. The curation principle is leaf
primitives only; composition is left LLM-side, on the empirical hypothesis
that frontier models compose tool calls more flexibly than any pre-baked
orchestration layer would. We verified this hypothesis live before the v2
release: a Claude Desktop session asked \emph{``check the health of the
eth\_dai\_v2 pool''}, the MCP server dispatched \texttt{CheckPoolHealth}
against a \texttt{MockProvider} twin, returned the typed result, and the
LLM correctly inferred the pool's state and proactively suggested
\texttt{DetectRugSignals} as a follow-up, without any explicit
composition primitive.

The MCP server itself is small (Listing~\ref{lst:mcp-server-sketch}); the
heavy lifting is in the primitive interface and the twin abstraction, both
of which exist independently of MCP. This is the substrate-not-product
principle: the agentic layer is a thin adapter, and the layer underneath
remains stable across whatever tool-calling protocol comes next.

\begin{lstlisting}[label={lst:mcp-server-sketch},caption={Sketch of the MCP server's tool dispatch. See \texttt{python/mcp/defipy\_mcp\_server.py} in the DeFiPy repository.}]
from defipy.tools import get_tool_schema, dispatch
from defipy.twin   import MockProvider, StateTwinBuilder

@server.list_tools()
async def list_tools():
    return [get_tool_schema(name) for name in CURATED_TOOLS]

@server.call_tool()
async def call_tool(name: str, args: dict):
    pool_id  = args.pop("pool_id")
    snapshot = MockProvider().snapshot(pool_id)
    lp       = StateTwinBuilder().build(snapshot)
    result   = dispatch(name, lp=lp, **args)   # primitive call
    return result.to_dict()                    # typed dataclass -> JSON
\end{lstlisting}

\subsection{Three Audiences, One Surface}

The most consequential property of the architecture is that one primitive
serves three audiences without modification:

\begin{itemize}
\item A \textbf{notebook quant} calls
      \texttt{AnalyzePosition().apply(lp, \dots)} against a synthetic
      twin and reads \texttt{result.il\_percentage}.

\item A \textbf{backtest} calls the same primitive against a
      \texttt{LiveProvider} twin pinned to a historical block, iterating
      the same call across thousands of blocks.

\item An \textbf{LLM agent} calls the same primitive via the MCP server,
      receives the same typed result serialized to JSON, and reasons over
      the named fields.
\end{itemize}

There is no per-audience adapter, no agent-specific path, no
notebook-specific shortcut. The primitive interface is the foundation; the
State Twin is the substrate; everything else composes on top.

% ============================================================
\section{Fork-and-Evaluate: Multi-Scenario Reasoning}
\label{sec:fork-evaluate}
% ============================================================

The State Twin's architectural payoff lands in the fork-and-evaluate pattern: pull live state once, fork the in-memory twin $N$ ways under distinct scenarios, run primitives against each fork, aggregate into a recommendation. This is the workflow the reactive architecture cannot support at all. By Remark~\ref{rem:factorization}, every state-dependent query in a reactive architecture factors through the chain-read map $\chi$. Hypothetical inputs $u'$ do not correspond to any value of $\chi(p)$ for any pool identifier $p$, so fork-and-evaluate cannot factor through $\chi$ and therefore cannot be evaluated reactively at any speed.

Concretely, forking on-chain state under a hypothetical scenario would require either $N$ real transactions (impossible for hypothetical scenarios) or $N$ RPC reads at $N$ hypothetical blocks (which do not exist). The substrate-mediated workflow described below is the realization of Remark~\ref{rem:factorization}'s corollary on the State Twin: each fork is an instance of the $\mathcal{X}$-valued object whose existence the corollary requires.

\subsection{The Pattern}

Listing~\ref{lst:fork-evaluate} is the complete substrate side of the
pattern. The chain read happens once at the top; every subsequent
scenario evaluation is local computation against an independent fork.

\begin{lstlisting}[label={lst:fork-evaluate},caption={Fork-and-evaluate over a live V3 pool. One chain read, $N$ in-memory twin forks, sub-second wall-clock for the entire scenario sweep.}]
import copy
from defipy.twin           import LiveProvider, StateTwinBuilder
from defipy.primitives.position import SimulatePriceMove

# 1. ONE chain read seeds the substrate
provider = LiveProvider(rpc_url)
snap = provider.snapshot(
    "uniswap_v3:0x88e6A0c2dDD26FEEb64F039a2c41296FcB3f5640"
)
lp = StateTwinBuilder().build(snap)   # USDC/WETH 5bps, pinned block

# 2. N independent in-memory forks under distinct price scenarios
scenarios = [-0.30, -0.20, -0.10, 0.0, +0.10, +0.20, +0.30]
results   = []
for pct in scenarios:
    fork = copy.deepcopy(lp)          # T4: independent twin per scenario
    res  = SimulatePriceMove().apply(
        fork,
        price_change_pct = pct,
        position_size_lp = 1.0,
        lwr_tick = snap.lwr_tick,
        upr_tick = snap.upr_tick,
    )
    results.append(res)
\end{lstlisting}

The semantics correspond directly to property T4 of
Definition~\ref{def:state-twin}: each \texttt{deepcopy} produces a twin
that evolves independently under its scenario without back-propagating to
the original. The state-space view in Section~\ref{sec:state-space} makes
this rigorous: each fork is a distinct trajectory $x_k \to x_{k+1}^{(s)}
= f(x_k, u^{(s)}, w^{(s)})$ under scenario $s$, and observation
$y^{(s)} = h(x_{k+1}^{(s)})$ is a structured dataclass the caller
can aggregate however they wish. By Proposition~\ref{prop:fidelity}, each
fork's trajectory remains within $\varepsilon$-fidelity of the
counterfactual on-chain trajectory it represents, with $\varepsilon$
controlled by the bound \eqref{eq:fidelity-bound} for the constant-product
family.

\subsection{Empirical Cost}

In the v2.1 reference demo \cite{DeFiPy26-Demo} against the USDC/WETH 5bps
V3 mainnet pool, $N{=}50$ scenarios complete in well under one second of
wall-clock time after the initial chain read. The dominant cost is the
single RPC round trip; every additional fork is essentially free,
modulo the cost of \texttt{deepcopy} on a small typed object. This is the
performance signature that distinguishes substrate-mediated reasoning
from RPC-mediated reasoning: $N$ scenarios at the cost of $1$ chain
read, rather than $N$ chain reads.

\subsection{What This Unlocks}

Fork-and-evaluate is the smallest non-trivial substrate-native workflow,
and it generalizes:

\begin{itemize}
\item \textbf{Distributional risk:} replace the deterministic price grid
      with samples from an empirical or model-implied distribution; the
      output is a true distribution over $y$, not a point estimate. The
      stochastic generalization in which $u_k$ and $w_k$ are promoted to
      random variables, and the substrate hosts Monte Carlo, filtering,
      or Bayesian primitives over the deterministic twin, composes
      naturally with the stochastic state-space treatments of AMMs in
      \cite{NKV24}, and is the subject of forthcoming work.

\item \textbf{Multi-agent ensemble reasoning:} spawn $N$ independent
      agents on $N$ forks, each with its own scenario-generation policy,
      objective function, or analytical specialization, and aggregate
      recommendations by ensemble methods (voting, stacking, Bayesian
      averaging). Property T4 is exactly the structural independence
      condition that ensemble theory requires; reactive architectures
      cannot host this workflow at all, because they offer one trajectory
      per epoch. We treat the agent-ensemble line of work as a separate
      forthcoming paper building on the substrate.

\item \textbf{Sequential reasoning:} chain forks across blocks, with
      $\widehat{x}_{k+1}$ from $f(\widehat{x}_k, u_k, w_k)$,
      to evaluate multi-step strategies without on-chain settlement.

\item \textbf{LLM-driven scenario design:} let the agent propose its own
      scenarios as tool inputs; the substrate evaluates them. The agent
      is no longer constrained to scenarios the developer pre-coded.

\item \textbf{Backtesting:} pin to a historical block, evaluate the same
      primitive across the historical trajectory, and the substrate
      becomes a backtest harness without any new infrastructure.
\end{itemize}

In each case, the State Twin is doing the same job: presenting an
in-memory, typed, replayable replica of pool state that the reasoning
layer (human or LLM) can fork freely against.

% ============================================================
\section{Discussion}
\label{sec:discussion}
% ============================================================

\subsection{Substrate, Not Product}
\label{sec:disc-substrate}

DeFiPy v2 deliberately does not ship an agent. It ships substrate: typed
primitives, the State Twin abstraction, an MCP adapter. The opinionated
choice is that the durable layer in agentic DeFi is the substrate, and
that wrapping the substrate in a particular agent framework today would
couple longevity to whichever framework happens to dominate next year.
The substrate-and-product separation echoes \texttt{ethers.js} and
\texttt{web3.py}, libraries whose value is precisely that they do not
opine on what the consumer is building.

\subsection{From Point-Estimate to Ensemble Reasoning}
\label{sec:disc-ensemble}

The substrate enables a structural transition in how DeFi protocols are
reasoned about. Human-mediated DeFi analysis is inherently a
point-estimate workflow: an analyst examines current chain state,
considers a small number of scenarios implicitly, and arrives at a
single judgment, with the cardinality of scenarios bounded by human
cognitive capacity. Machine-mediated DeFi analysis over the State Twin
substrate is inherently an ensemble workflow: the agent forks the
substrate $N$ ways, evaluates $N$ scenarios in parallel, and reasons
over the resulting distribution of outcomes, with the cardinality
bounded only by compute. This transition is not a change of speed; it
is the structural shift from point-estimate decision-making to ensemble
reasoning in the sense familiar from Bayesian model averaging, bootstrap
aggregation, and Monte Carlo simulation. The state-space framework
adopted in Section~\ref{sec:state-space} is what makes this transition
formalizable: trajectories are first-class objects, ensembles of
trajectories compose naturally, and per-trajectory observations
aggregate into distributions over outcomes.

The operational consequence is a corresponding shift in how DeFi capital
allocation decisions are made. An LP evaluating position risk under
price movement is fundamentally asking a distributional question, but
the reactive architecture forces the question into a point-estimate
workflow, producing decisions grounded in single scenarios rather than
in the distribution the LP actually faces. The substrate corrects this
mismatch by making distributional analysis the natural workflow: an LP
can compute value-at-risk, expected shortfall, or probability-of-loss
across thousands of price trajectories in sub-second wall-clock time
after a single chain read. Protocol comparison becomes apples-to-apples
because the same scenario can be evaluated against V2, V3, Balancer, and
Stableswap twins under identical inputs. Backtesting becomes
structurally trustworthy because historical replay and counterfactual
simulation are unified under the same provider/builder factorization.
Multi-pool joint optimization becomes accessible because forking and
evolving distinct pool states under shared price inputs is what the
substrate is structurally designed to do.

Beyond the change in decision quality, the substrate changes the
decision's accountability structure. An ensemble-mediated decision over
named primitives producing typed dataclasses is reconstructable in a way
that intuition-mediated decisions over point estimates are not: the
scenario set, the primitive invocations, and the typed outputs together
form an auditable trace of the reasoning that produced the decision.  This is the property that makes the substrate composable with the provenance layer and the trust-gate role discussed in
Section~\ref{sec:related}: the substrate produces ensemble-based reasoning; the provenance layer anchors the resulting decisions; a trust gate, when formalized, would enforce abstention when the ensemble's confidence is insufficient. Each concern is distinct, and together they form the architecture for agentic DeFi reasoning that is faster, distributionally honest, and auditable in a way reactive architectures cannot be. The substrate itself remains deterministic. The fidelity bound of Proposition~\ref{prop:fidelity} is worst-case-deterministic. But the workflow it enables is fundamentally ensemble-based, because the agent
operates on distributions over trajectories rather than on single trajectories. The forthcoming stochastic extension formalizes the ensemble explicitly by promoting $u_k$ and $w_k$ to random variables; the present paper establishes the deterministic substrate on which that
formalization composes.

\subsection{Read-Only by Design}
\label{sec:disc-readonly}

\texttt{LiveProvider} is read-only. There is no \texttt{provider.sign()},
no \texttt{provider.send\_transaction()}, no embedded key management.
Signing infrastructure varies enormously across consumers (local key,
hardware wallet, MPC vault, signing service, custodial flow), and the
substrate is not the right place to opine on any of them. The escape
hatch \texttt{provider.get\_w3()} returns the underlying \texttt{web3.Web3}
instance for consumers who need to act after analysis; from there, their
signing infrastructure takes over. The substrate stops at the chain-read
boundary on purpose.

\subsection{Math Stays LLM-Free}
\label{sec:disc-llm-free}

The math layer (the 21 analytical primitives, the State Twin builder, the
protocol invariants) has zero LLM dependencies and zero web3
dependencies in the core install. This is what allows the math to be
auditable, reproducible, and trustworthy at the substrate level. The
LLM integration lives strictly at the MCP adapter layer; remove it and
the rest of the library is unaffected. We regard this discipline as
non-negotiable: the moment the math layer takes an LLM dependency, the
whole stack becomes only as deterministic as the LLM, which is to say,
not deterministic at all.

\subsection{Generalization Beyond DeFi}
\label{sec:disc-generalization}

The State Twin pattern (typed off-chain replica, provider/builder
split, primitive interface, agentic adapter) is not specific to DeFi.
Any domain in which (a) the underlying system has a well-defined
state-transition map, (b) the system's canonical state is expensive to
query, and (c) reasoning over counterfactuals is high-value, is a
candidate for the same architecture. Manufacturing process control,
clinical decision support, supply-chain logistics, and grid operations
all fit this template. We expect the substrate-not-product framing to
travel; the demonstration that it lands cleanly in DeFi is the easy case,
not the hard one.

\subsection{Limitations and Future Work}
\label{sec:disc-limitations}

Several limitations frame the current substrate. \texttt{LiveProvider} in
v2.1 ships V2 and V3 only; Balancer and Stableswap live readers are v2.2
work. V3 \emph{active-liquidity} reads are full-fidelity; V3
\emph{cross-tick} primitives that require walking the tick bitmap are
also v2.2. The substrate is currently single-pool; multi-pool routing
substrate (treating a graph of pools as a composite twin) is a natural
extension but not yet implemented. The fidelity bound in
Proposition~\ref{prop:fidelity} is currently stated for the
constant-product family; analogous explicit bounds for Balancer's
weighted-geometric-mean and Stableswap's amplified invariant under their
respective fixed-point implementations are a clean open problem.
Finally, we have not formalized the trust gate that should sit between
an LLM proposing a transition $u_k$ and the substrate evaluating it;
this is an open problem and the next major piece of the agentic-DeFi
safety story.

\subsection{Why Reactive Is Not Just Slower}
\label{sec:disc-reactive-structural}

A natural objection to the substrate framing is that reactive
architectures are merely slower, not structurally limited, and that
sufficient engineering (better RPC infrastructure, parallel chain reads,
optimized caching, fork-node access) closes the gap.
Remark~\ref{rem:factorization} makes explicit why this objection fails:
the gap is not a performance gap, it is a categorical one. A reactive
architecture's state-dependent queries are forced to factor through the
chain-read map $\chi$, which means counterfactual transitions, ensemble
queries, and trajectory rollouts are not supported \emph{at any speed},
because the queries themselves do not correspond to values of $\chi$
for any pool identifier. Faster RPC infrastructure makes the supported
queries faster; it does not enlarge the set of supported queries. The
substrate's contribution is not to accelerate reactive workflows but to
unlock a strictly larger query class that reactive architectures cannot
host by construction.

The practical implication is that the substrate-versus-reactive question
is settled at the architectural level rather than the engineering level:
the choice is not between two implementations of the same capabilities,
but between two query classes of different cardinality. Once an agentic
workflow requires counterfactual reasoning (e.g., \emph{what is my
position value under a $-30\%$ price shock?}, \emph{which of five
candidate rebalances is optimal?}, \emph{what is the distribution of
outcomes across a thousand price trajectories?}), it has crossed a
structural boundary that no optimization of the reactive path can
recover. Anchored-state alternatives (forked nodes, archive RPCs, cached
snapshots) are special cases of the substrate when they exist as typed
in-memory objects admitting direct $h$ and $f$ evaluation; they are
\emph{not} faster reactive paths but partial realizations of the
substrate framing under a different name.

% ============================================================
\section{Related Work}
\label{sec:related}
% ============================================================

\paragraph{Digital twins.} The digital-twin concept originated in
manufacturing and aerospace contexts \cite{Gri17}, where a real-time
in-memory replica of a physical asset enables monitoring and predictive
maintenance. The State Twin imports the core idea, a high-fidelity
in-memory replica that supports counterfactual analysis, into the
DeFi setting, with the additional property that the replica's fidelity
to its source is exact in the rational-arithmetic case and tightly
bounded under fixed-point arithmetic
(Proposition~\ref{prop:fidelity}, \cite{TG25}), owing to the determinism
of the underlying state-transition map.

\paragraph{Automated market makers.} The mathematical structure of AMMs
is well-studied. The constant-product invariant of Uniswap V2
\cite{Adams20}, the concentrated-liquidity construction of V3
\cite{Adams21}, the weighted-pool invariant of Balancer \cite{Mar19}, and
the amplified invariant of Curve \cite{Ego19} together form the core
formal foundation. The state-space framing
(Section~\ref{sec:state-space}) is a re-presentation of these protocols
as discrete-time controlled dynamical systems; this framing follows
straightforwardly from the protocol invariants and is used here as
substrate for the State Twin abstraction rather than as a contribution
in its own right.

\paragraph{Formal-methods treatments of AMMs.} Tranquilli and Gupta
\cite{TG25} formalize Uniswap v3 as networks of priced timed automata
and finite-state transducers, prove an explicit additive rounding bound
for the discretized constant-product update, and verify the bound
exhaustively on small TLA+ instances via TLC. Their lemma underwrites
Proposition~\ref{prop:fidelity} above, which is the quantitative form of
property T2 in Definition~\ref{def:state-twin}. The directions are
complementary: their work targets exhaustive correctness verification
of pool dynamics, ours targets agentic reasoning over an off-chain typed
substrate that inherits their bound. Earlier theorem-proving work on
constant-product AMMs in Lean~4 \cite{PB24} addresses correctness in a
similar formal-methods spirit but for the v2-style full-range
construction.

\paragraph{Stochastic state-space approaches to AMMs.} Nadkarni,
Kulkarni, and Viswanath \cite{NKV24} formalize AMM-mediated price
discovery as a hidden-state Markov model with the external price as the
latent variable and trader actions as noisy observations, and derive
Kalman-filter-based optimal adaptive bonding curves under Gaussian and
lognormal price processes. Their work uses the stochastic state-space
view in service of mechanism design (redesigning the AMM curve itself
to minimize liquidity-provider loss), whereas the present paper uses a
closely related deterministic state-space view in service of substrate
construction. The stochastic generalization of the State Twin substrate,
in which $u_k$ and $w_k$ in \eqref{eq:state-space} are promoted to
random variables and the fork-and-evaluate pattern becomes Monte Carlo,
filtering, or Bayesian inference over the deterministic twin, is the
natural composition with \cite{NKV24} and is the subject of forthcoming
work. We deliberately keep the substrate itself deterministic in both
the present paper and that follow-up: what changes in the stochastic
extension is the inputs and disturbances the substrate consumes and the
primitives that run against it, not the substrate's underlying
state-transition map.

\paragraph{Agentic DeFi systems.} A growing body of work explores LLM-
and agent-driven interaction with DeFi protocols, primarily through
the lens of natural-language interfaces over chain state. The Model Context Protocol \cite{MCP24}, on which the agentic exposure of DeFiPy v2 is built, is a recent industry standard for LLM-tool integration. The present paper identifies the reactive
architecture as a structural problem and proposes an off-chain typed
substrate as the fix; prior work on agentic DeFi has focused
primarily on natural-language interfaces over chain state rather than
on the substrate layer underneath.

\paragraph{Backtesting and simulation frameworks.} Frameworks for
backtesting trading strategies on AMMs exist as research codebases and
proprietary systems. These typically conflate \emph{historical replay}
with \emph{counterfactual simulation} and do not separate state source
from state shape; the State Twin's provider/builder split treats both
uniformly, with synthetic, live, and historical sources interchangeable
behind a single primitive surface.

\paragraph{Provenance for agent actions.} Earlier work on
AnchorRegistry \cite{Moo26-AR} formalizes a cryptographic-priority
mechanism for provenance attribution under operator gating, complementary
to the State Twin's substrate-mediated reasoning. A formal trust gate
sitting between an LLM and any state-modifying action remains an open
problem: the substrate provides safe \emph{reasoning} and the provenance
layer provides \emph{auditable history}, but the action gate that
enforces abstention under insufficient confidence is not yet formalized.
Formalizing the action gate and composing it with the substrate and provenance layers is the next natural piece of work.

% ============================================================
\section{Conclusion}
\label{sec:conclusion}
% ============================================================

We have argued that the dominant reactive architecture for agentic DeFi
systems is structurally limited, and that the missing layer is an
off-chain typed substrate that preserves the protocol's exact mathematics
while admitting the operations on-chain state cannot: forking, replay,
branching, counterfactual rollout. We named this substrate the
\emph{State Twin}, formalized it through a state-space view of AMM
protocols, and gave a clean provider/builder factorization that
decouples state source from state shape. We proved a quantitative
\emph{State Twin Fidelity} bound on the divergence between twin and
on-chain trajectories under the discretized constant-product update, by
specializing the rounding lemma of Tranquilli and Gupta \cite{TG25} to
the substrate setting. We described its open-source realization in
DeFiPy v2 and demonstrated the substrate empirically through a
fork-and-evaluate workflow that evaluates fifty distinct price-shock
scenarios in sub-second wall-clock time after a single chain read.

The contribution is the substrate, not a particular agent. The State
Twin pattern is a precise specification of
\emph{what an agentic DeFi substrate must look like}: typed,
fidelity-bounded, forkable, source-agnostic, primitive-clean,
LLM-native at the adapter layer. We expect the pattern to generalize
beyond DeFi to any domain where reasoning over counterfactuals against
a deterministic state-transition map has economic or operational value.

% ============================================================
\section*{Acknowledgments}
% ============================================================

The author thanks the open-source contributors to UniswapPy, BalancerPy,
StableswapPy, and Web3Scout, whose protocol-specific implementations
realize the hand-derived AMM math of \cite{Moo25-Book} and underpin the
State Twin builder dispatch. AI writing assistance
(Anthropic Claude) was used in drafting this manuscript; all theoretical
contributions, formal definitions, system-design decisions, and
architectural choices are the author's own. The reference implementation
described in this paper is available at
\url{https://github.com/defipy-devs/defipy} and on PyPI as
\texttt{defipy}, with full documentation at \url{https://defipy.org}.

% ------------------------------------------------------------
% Provenance anchor for this manuscript
% ------------------------------------------------------------
\vspace{1em}
\noindent\textbf{Provenance anchor.} A SHA-256 manifest of this
manuscript is anchored at \url{https://anchorregistry.com} for priority
attestation.

% ============================================================

\end{document}